\documentclass[11pt]{article} 
\usepackage[stable]{footmisc}
\usepackage{amsfonts}
\usepackage{amssymb,amsmath}
\usepackage{amsthm,amsgen}
\usepackage{graphicx,epsfig}
\usepackage[mathscr]{eucal}
\usepackage{slashed}
\usepackage{bbm}
\usepackage{ marvosym }
\usepackage{dsfont}
\usepackage{accents}
\usepackage{hyperref}
\usepackage{color}
\newtheorem{thm}{THEOREM}[section]
\newtheorem*{thm*}{THEOREM}

\newtheorem{prop}[thm]{Proposition}
\newtheorem*{prop*}{Proposition}
\theoremstyle{definition}

\newtheorem{rem}[thm]{Remark}

\DeclareMathOperator{\tr}{tr}

\newcommand{\refeq}[1]{{\rm(\ref{#1})}}
\renewcommand{\eqref}[1]{{\rm(\ref{#1})}}

\newcommand{\vect}[1] {\boldsymbol{{ #1}} }

\newcommand{\bpi}{\boldsymbol{\pi}}

\newcommand{\fV}{\vect{f}}              
\numberwithin{equation}{section}

\newcounter{subequation}
	\newenvironment{subequation}%
	{\addtocounter{equation}{-1}%
	\stepcounter{subequation}%
	\begin{equation}}%
	{\end{equation}%
}

\newcommand{\fa}{\mathfrak{a}}

\newcommand{\bb}{\mathbf{b}}
\newcommand{\ba}{\mathbf{a}}
\newcommand{\bfa}{\boldsymbol{\fa}}

\newcommand{\bc}{\mathbf{c}}

\newcommand{\be}{\mathbf{e}}

\newcommand{\beff}{\mathbf{f}}

\newcommand{\bJ}{\mathbf{J}}

\newcommand{\bK}{\mathbf{K}}

\newcommand{\bL}{\mathbf{L}}

\newcommand{\bn}{\mathbf{n}}

\newcommand{\bP}{\mathbf{P}}

\newcommand{\bs}{\mathbf{s}}
\newcommand{\bT}{\mathbf{T}}

\newcommand{\bu}{\mathbf{u}}

\newcommand{\bX}{\mathbf{X}}
\newcommand{\bx}{\mathbf{x}}
\newcommand{\bY}{\mathbf{Y}}

\newcommand{\btau}{\boldsymbol{\tau}}

\newcommand{\beq}{\begin{equation}}
\newcommand{\eeq}{\end{equation}}
\newcommand{\bseq}{\begin{subequation}}
\newcommand{\eseq}{\end{subequation}}

\newcommand{\Id}{\mathbbm{1}}

\newcommand{\p}{\partial}

\newcommand{\cB}{\mathcal{B}}
\newcommand{\cC}{\mathcal{C}}

\newcommand{\cL}{{\mathcal L}}
\newcommand{\cLph}{{\mathbf L}_{\mbox{\tiny{ph}}}}

\newcommand{\cM}{{\mathcal M}}

\newcommand{\M}{\mathcal{M}}

\newcommand{\muV}{\mu}
\newcommand{\siV}{\boldsymbol{\sigma}}

\newcommand{\bLa}{\boldsymbol{\Lambda}}

\newcommand{\psiPH}{{\psi_{\mbox{\rm\tiny ph}}}}
\newcommand{\psiPHb}{{\overline\psi_{\mbox{\rm\tiny ph}}}}

\newcommand{\psiEL}{{\psi_{\mbox{\rm\tiny el}}}}

\newcommand{\phiPH}{{\phi_{\mbox{\rm\tiny ph}}}}

\newcommand{\psib}{\overline{\psi}}

\newcommand{\mPH}{{m_{\mbox{\Lightning}}}}

\newcommand{\Aset}{{\mathbb A}}

\newcommand{\Cset}{{\mathbb C}}

\newcommand{\Pset}{\mathbb{P}}

\newcommand{\Rset}{{\mathbb R}}

\newcommand{\la}{\lambda}

\newcommand{\De}{\Delta}
\newcommand{\al}{\alpha}

\newcommand{\ga}{\gamma}

\newcommand{\bga}{\boldsymbol{\gamma}}
\newcommand{\ep}{\epsilon}

\newcommand{\ka}{\kappa}

\newcommand{\si}{\sigma}

\newcommand{\Si}{\Sigma}
\newcommand{\nab}{\nabla}
\newcommand{\half}{\frac{1}{2}}

\newcommand{\diag}{\mbox{diag}}
\newcommand{\Span}{\mbox{Span}}

\newcommand{\bna}{\begin{eqnarray}}
\newcommand{\ena}{\end{eqnarray}}
\newcommand{\bea}{\begin{eqnarray*}}
\newcommand{\eea}{\end{eqnarray*}}
\newcommand{\ben}{\begin{enumerate}}
\newcommand{\een}{\end{enumerate}}
\newcommand{\bi}{\begin{itemize}}
\newcommand{\ei}{\end{itemize}}

\newcommand{\Hset}{{\mathbb H}}

\newcommand{\dal}{\raisebox{3pt}{\fbox{}}\,}

\title{Noetherian Conservation Laws for Photons}
\author{\normalsize\sc{Michael K.-H. Kiessling and A. Shadi Tahvildar-Zadeh}\\
	{$\phantom{nix}$}\\[-0.1cm] 
        \normalsize Department of Mathematics\\[-0.1cm]
	Rutgers, The State University of New Jersey (New Brunswick)
}
\date{\today\\ {\em Dedicated to Tom Sideris on the occasion of his 70th birthday}} 

\begin{document}
\maketitle
\begin{abstract} 
We review the formulation of a Lorentz-covariant bispinorial wave function and wave equation for a single photon on a flat background. We show the existence of a 10-dimensional set of conservation laws for this equation, and prove that 8 of these can be 
{used to 
obtain} global, gauge-invariant, ADM-like quantities that together define a covariantly constant self-dual bispinor. 
\end{abstract}
\section{Introduction}
Of central importance to the modern development of the theory of partial differential equations has been the notion of {\em conservation laws} and their connection with continuous symmetries of the equations, as developed in utmost generality by E. Noether in her two celebrated theorems of 1918 \cite{Noe1918}. Noether's two theorems are so general and all-encompassing, and in their theoretical significance so ahead of their time that, more than fifty years after their discovery, manuscripts were still being submitted and published that purported to be ``a generalization of Noether's theorem", while in reality they were special cases\footnote{For a fascinating account of the reception of Noether's theorems by the mathematics and physics communities, and their influence on the development of these sciences, see \cite{YKS2018}.} \cite{Olv1968}.    The search for such conservation laws, whether exact or approximate, whether local or global or even {\em microlocal}, and the {\em a priori} estimates they imply for solutions of hyperbolic (or more generally, evolutionary) partial differential equations, have been an {\em id\'ee fixe} in the works of many of the giants of this field, chief among them F. John, C. Morawetz, W. Strauss, S. Klainerman, D. Christodoulou, J. Shatah, and T. C. Sideris.    

Indeed, a recurrent theme in many of Sideris's seminal works in hyperbolic PDEs is how to compensate for the {\em lack} of a particular symmetry in a system, whose associated conservation law --had it been available-- would have made the deriving of an {\em a priori} estimate straightforward by appealing to previously known results.  The case in point is the pioneering work of  Klainerman \cite{Kla1985} in which he obtained the dispersive estimates that are now named after him, for solutions of the linear wave equation on a flat background.  Klainerman made use of local conservation laws associated with the full 15-dimensional  group of conformal isometries of $3+1$ dimensional Minkowski space, which is generated by 4 translations, 3 spatial rotations, 3 boosts, 4 inverted translations, and one scaling vectorfield.  He used them all because they were all there, not because they were all absolutely necessary.  Sideris was among the first to realize that one could get away with fewer symmetries.  In his work on the Klein-Gordon equation \cite{Sid89} he noted that the scaling is no longer available, and showed that boosts can be used instead.  Later, in his collaboration with Klainerman on equations of elastodynamics \cite{KlaSid}, the reverse was the case: boosts were not available, so scaling had to be used, combined with an ingenious use of weighted norms.  This observation in turn became the cornerstone in Sideris's celebrated works on nonlinear elastodynamics \cite{SidInvent} and \cite{SidAnnals}, published in {\em Inventiones} and the {\em Annals} respectively, as well as the path-breaking work on systems of wave equations with multiple speeds \cite{SidMult}.  

It is worth noting that, the oft-repeated phrase about ``Noether's Theorem"\footnote{When used in singular form, what is meant most often is Noether's First Theorem.} being a result in classical field theory stating that ``isometries of spacetime give you conservation laws, for Lagrangians depending on the fields and their first derivatives" is but a cartoon version of Noether's actual result, which concerns itself with systems of partial differential equations in any number of dependent and independent variables and any number of derivatives, without the need for the existence of a metric for the domain or any notion of isometry. The notion of ``symmetry" developed by Noether was likewise completely general, allowing for symmetries acting only on the independent variables (i.e. symmetries of the domain), or only on the dependent variables (so-called target symmetries), as well as those with a combined action on the Cartesian product of the domain and the target. The latter play a key role in Sideris's works on elastodynamics, where the dynamical object is the vector-valued mapping between the ``undeformed" reference state of a solid and its ``deformed" state in physical space.  Thus Sideris's extensive use of {\em simultaneous rotations} is an instance of such ``product" Noether currents.

Another key situation arises in tensorial equations, in which the target space is a vector or principal bundle over the domain manifold. This is the realm of Noether's Second Theorem, which concerns itself with what nowadays are called ``gauge symmetries."  These lie at the heart of today's fundamental physical theories such as the Standard Model of particle physics. For a careful expository account of Noether's work using modern terminology that includes the important case of Lagrangian theory for sections of vector bundles, see Christodulou's beautiful monograph \cite{Chr2000}.

It is a fair question to ask, what to do when there are {\em no} symmetries available?  This is most famously the case when one studies Einstein's equations of general relativity as a system of PDEs, where the unknown is the spacetime itself, which --when not vacuum-- has no a priori reason to have any symmetries.  For the important subset of {\em asymptotically flat} spacetimes, however, there is a notion of symmetry available ``at infinity."  It was already known to Einstein \cite{Ein1916} and Weyl \cite{WeyRZM}, and later Landau and Lifschitz \cite{LanLifCFT}, that such spacetimes will have {\em global} conserved quantities, and it was shown by Arnowitt, Deser and Misner \cite{ADM} that these are in the form of integrals over the ``sphere at infinity'' of a spatial slice, and that they are independent of the asymptotic time variable.  For a rigorous derivation of the ADM conservation laws for Einstein's equations using Noether's Theorems see \cite{MPGR1}. 

Our goal in this paper was to obtain a set of gauge-independent conservation laws for a particular quantum mechanical equation that we have previously proposed \cite{KTZ2017} as the relativistic equation of evolution for the wave function of a single photon, i.e. a massless spin-one particle with no longitudinal modes. Our hope was to obtain these on a spacetime background with no assumed symmetries. {One can see that the requirement of gauge-independence limits one at best to global quantities that can be defined on an asymptotically flat spacetime}, in analogy to the ADM quantities of General Relativity. Further restrictions coming from the need for the existence of covariantly constant tensorfields at the moment seem to limit our conclusions to flat spacetimes, {yet we end up with new global conserved quantities for our photon wave equation on} Minkowski space! We believe this is further evidence that, nearly 100 years after their discovery, Noether's Theorems continue to remain relevant to mathematical physics and be a source of inspiration (and surprise) to its practitioners, just as they have been to Sideris and many others before him. 
\section{Review of Photon Wave Function and Equation on Minkowski Space}
\subsection{Algebraic Preliminaries}
\subsubsection{The Clifford algebra associated with a spacetime} 
Of central importance to the relativistic formulation of quantum mechanics in $d$ space dimensions is the {\em  spacetime algebra} $\Aset$, defined as the complexification of the real Clifford algebra $\mbox{Cl}_{1,d}(\Rset)$ associated with the Minkowski quadratic form of signature $(+,-,\dots,-)$.  For $d=3$ this algebra is easily seen to be isomorphic to the algebra of $4\times 4$ complex matrices, as well as the algebra of $2\times 2$ matrices with complex quaternionic entries:
$$\Aset := \mbox{Cl}_{1,3}(\Rset)_\Cset \cong M_4(\Cset) \cong M_2(\Pset) \cong \Pset \otimes \Pset,$$
where $\Pset := \Cset \otimes \Hset$ is the algebra of complex quaternions, also known as the Pauli algebra \cite{HestenesBook}.  (All tensor products are over $\Cset$ unless otherwise noted.)

 The isomorphism of $\Aset$ with $M_4(\Cset)$ can be realized by choosing a basis for $\mbox{Cl}_{1,3}(\Rset)$ and taking complex linear  combinations of the basis elements. 
 A convenient basis for the real algebra is formed  by the so-called Dirac gamma-matrices (in their Weyl representation) and their products: 
 Let $\mathds{1}_n$ denote the $n\times n$ identity matrix, and define
\beq\label{def:gammas}
\ga^0 = \left( \begin{array}{cc}0 & \mathds{1}_2 \\ 
\mathds{1}_2 & 0 \end{array} \right),\qquad \ga^k = \left( \begin{array}{cc}0 & -\si_k \\
 \si_k & 0 \end{array} \right),\ k=1,2,3,
\eeq
where the $\si_{k\in\{1,2,3\}}$ are the three conventional Pauli matrices. 
 The $\ga$-matrices satisfy the Clifford algebra relations
\beq \label{gammasON}
\ga^\mu \ga^\nu + \ga^\nu \ga^\mu = 2 \eta^{\mu\nu}\Id_4;
\eeq
where the $\eta^{\mu\nu}$ are the components of the Minkowski metric tensor 
\beq
\boldsymbol{\eta} = \diag(1,-1,-1,-1).
\eeq 

Any Clifford algebra $A$ associated with a vector space $V$ over a field $F$ contains a subspace that is isomorphic to $V$. 
 The elements of that subspace are called {\em 1-vectors}.  
Note that \refeq{gammasON} simply states that the four matrices $\{\ga^\mu\}_{\mu=0}^3$ form a Lorentz-orthonormal set of 
1-vectors in the Clifford algebra $\mbox{Cl}_{1,3}(\Rset)$.  

 By definition, a {\em $k$-vector} is the (Clifford) product of $k$ elements, each one of which is a 1-vector. 
 Let $\{\be_j\}_{j=1}^n$ be a basis for $V$. Every Clifford number $a\in A$ has a $k$-vector expansion of the form 
\beq a = a_S\Id + \sum_{k=1}^{n}\sum_{1\leq i_1<\dots<i_k\leq n} a^{i_1\dots i_k} \be_{i_1}\dots\be_{i_k},
\eeq
where $n=\dim_F V$ and the coefficients $a_S,a^{i_1\dots i_k}$ are in $F$.  
 It follows that the following is a basis for the (16-dimensional) algebra $\mbox{Cl}_{1,3}(\Rset)$:
\beq\label{basis}
B := \left\{ \mathds{1}_4; \ga^0, \ga^1, \ga^2,\ga^3; \ga^0\ga^1, \ga^0\ga^2,\dots; \ga^0\ga^1\ga^2, \dots;\ga^0\ga^1\ga^2\ga^3 \right\},
\eeq
and therefore its complexification can be obtained by taking the coefficients of the expansion to be complex numbers: $\Aset = \Span_\Cset B$.  

 The complexified algebra $\Aset$ in particular includes the {\em pseudoscalar}
\beq 
\ga^5 := i \ga^0 \ga^1 \ga^2 \ga^3 
= \begin{pmatrix} \, \Id_2 & \ 0 \\ 0 & -\Id_2\end{pmatrix},
\eeq
and therefore the projections
\beq\label{def:projections}
\Pi_\pm := \half( \Id_4 \pm \ga^5).
\eeq
Using these projections it follows right away that $\Aset$ contains all $4\times 4$ matrices, 
and it is easy to verify that $\Aset$ and $M_4(\Cset)$ are indeed isomorphic as algebras, with the 
Clifford multiplication given by matrix multiplication.

 Let $a = a_S \Id +\sum_I a_I \bga^I$ denote the $k$-vector expansion of $a\in \Aset$. 
 Thus $a_S,a_I\in \Cset$ and each $\bga^I$ is a $k$-fold product of gamma matrices, for some $k$. 
 Two important operations on Clifford numbers are the following:  
\begin{itemize}
\item 
The {\em scalar part} $a_S$ of an element $a \in \Aset$ is by definition the coefficient of the unit element $\Id$ in the expansion of $a$ in 
any basis (such as $\cB$.)  
Using the isomorphism above, we can view $a$ as a $4\times 4$ matrix, and we then have
\beq \label{def:Spart}
a_S = \frac{1}{4}\tr a,
\eeq
where $\tr$ denote the usual operation of taking the trace of a matrix.
\item 
The {\em conjugate reversion} (a.k.a. {\em Dirac adjoint}) $\overline{a}$ of  $a\in\Aset$ is by definition the element 
obtained by reversing the order of multiplication of the 1-vectors in the expansion of $a$ in terms of $k$-vectors, 
and taking the complex conjugate of the coefficients in that expansion. 
 Thus
$
\overline{a} = a_S^* \Id + \sum_I a_I^* \tilde{\bga}^I 
$
with $\tilde{\bga}^I = \ga^{i_k}\dots\ga^{i_1}$ whenever $\bga^I = \ga^{i_1}\dots\ga^{i_k}$. 
 Using the isomorphism $\Aset \cong M_4(\Cset)$ it is not hard to see that, when $a$ is viewed as a $4\times 4$ matrix,
\beq\label{def:Diradj}
\overline{a} = \ga^0 a^\dag \ga^0,
\eeq
where $a^\dag = (a^*)^{\mathrm{T}}$ denotes the conjugate-transpose of $a$. 
(Here and elsewhere, ${}^*$ denotes complex conjugation $i\to -i$, while ${}^{\mathrm{T}}$ denotes the matrix transpose.)
\end{itemize}

The isomorphism $\Aset \cong \Pset \otimes \Pset$ on the other hand, can be realized by first noting that the complexification of the real Clifford algebra $\mbox{Cl}_{1,3}(\Rset)$ yields the complex Clifford algebra $\mbox{\bf Cl}(4)$, which is known to be isomorphic to $\mbox{\bf Cl}(2) \otimes \mbox{\bf Cl}(2)$, and that $\mbox{\bf Cl}(2)$ is in turn isomorphic to $\Pset$. 

Finally, since $\Pset \cong M_2(\Cset)$, by partitioning a $4\times 4$ matrix into four $2\times 2$ blocks in the obvious way, we may view an element  $\Phi \in \Aset \cong M_4(\Cset)$ as a $2\times 2$ matrix with entries in $\Pset$:
\beq\label{def:PhotonWF}
\Phi = \begin{pmatrix} \phi_+ & \chi_- \\ \chi_+ & \phi_- \end{pmatrix},\qquad \phi_\pm, \chi_\pm \in \Pset \cong M_2(\Cset).
\eeq
\subsubsection{Spinors and differential forms}
In \cite{KTZ2017} it was shown that there is a correspondence between rank-one spinors and 1-forms, and similarly between {\em trace-free} rank-two spinors and 2-forms.  Specifically, for every rank-two bispinor 
$\psiPH=\begin{pmatrix}
    \phi_+ & \chi_-\\ \chi_+ & \phi_-
\end{pmatrix} $ that satisfies the trace condition $\tr \phi_+ = \tr \phi_- = 0$, there exists two real-valued 2-form $\mathfrak{f}_\pm$ and two complex-valued 1-forms $\mathfrak{a}_\pm$ such that
$$ \phi_+ =\Sigma(\mathfrak{f}_+),\quad \phi_- = \Sigma'(\mathfrak{f}_-),\quad \chi_+ = \sigma'(\mathfrak{a}_+),\quad \chi_- = \sigma(\mathfrak{a}_-),$$
where the (invertible) mappings $\si,\si',\Si,\Si'$ are defined as follows: Let $\si_0 = \si'_0 = \Id_{2\times 2}$.  Let $\si_1,\si_2,\si_3$ denote the Pauli matrices, and let $\si'_k := - \si_k$ for $k=1,2,3$.  Then we define, for contravariant tensors $A$ and $F$
$$
\si(A) := \si_\mu A^\mu,\quad \si'(A) := \si'_\mu A^\mu,\quad \Si(F) := \frac{i}{4}\si_\mu {\si'}_\nu F^{\mu\nu},\quad \Si'(F) := \frac{i}{4}\si'_\mu {\si}_\nu F^{\mu\nu}.
$$
The action on other tensor types can then be defined via the musical isomorphism $\sharp$, i.e. raising of subscript indices using the (inverse) spacetime metric.

\subsubsection{Lorentz group $O(1,3)$ and its spinorial representation}
The group of spacetime rotations of $\Rset^{1,3}$ is the Lorentz group $O(1,3)$.  This group is disconnected: it has four connected components.  The connected component that contains the identity element is called the {\em proper} Lorentz group.  Viewed as a matrix group, the proper Lorentz group is identified with $SO_0(1,3)$, the subgroup of matrices $\bLa \in O(1,3)$ with  $\det \bLa = 1$ and $\bu^T \bLa \bu>0$ for all future-directed timelike vectors $\bu \in \Rset^{1,3}$.

The full Lorentz group $O(1,3)$ is generated by elements of $SO_0(1,3)$ together with the space-reflection
\begin{equation}\label{def:P}
    \bP := \left(\begin{array}{cc} 1 & 0 \\ 0 & -\mathds{1}_3 \end{array}\right)
\end{equation} 
and time-reversal $\bT := - \bP = \left(\begin{array}{cc} -1 & 0 \\ 0 & \mathds{1}_3 \end{array}\right)$.
For $\bx \in \Rset^{1,3}$ let 
\beq\label{def:gax}
\ga(\bx) := \ga_\mu x^\mu  = \begin{pmatrix}
0 & x^0\mathds{1}+x^i\si_i \\ x^0\mathds{1} - x^i\si_i & 0 
\end{pmatrix} = \begin{pmatrix} 0 & \si(\bx) \\ \si'(\bx) & 0 \end{pmatrix}
\in \Aset
\eeq
be the image of $\bx$ under the standard embedding of the Minkowski spacetime into its Clifford algebra (Note that indices are raised and lowered using the Minkowski metric $\eta$, and that we are using Einstein's summation convention).   It is a standard result of the representation theory of the Lorentz group that for every $\bLa\in SO_0(1,3)$ there exists $A = A(\bLa) \in SL(2,\Cset) \subset \Pset$ such that
$$
\ga(\bLa \bx) = L_{\bLa} \ga(\bx) L_{\bLa}^{-1}
$$
where 
\beq \label{def:Spingroup}
L_{\bLa} = \left(\begin{array}{cc}A & 0 \\ 0 & A^* \end{array}\right) \in \Aset.
\eeq
(Recall that since $\det A = 1$ we have $A^* = \tilde{A}^\dag = A^{-\dag}$.)

It is also easy to check that 
$$
\ga(\bP \bx) = \ga^0 \ga(\bx) \ga^0
$$
so that, as an operator on $\Cset^4$, we can set 
\beq\label{def:LP}
L_\bP = \ga^0.
\eeq
There are of course other choices for $L_\bP$, for example $L_\bP = i \ga^0$.  In fact we know that, since the Lorentz group is not connected, it does not have a unique covering group: it has eight non-isomorphic coverings.  In four of these, the representation of the time-reversal operator $\bT$ is unitary, and in the other four it is anti-unitary (see \cite{ThallerBOOK}, Thm. 3.~10.) Here we are going to make a convenient choice for $L_\bT$ that makes it anti-unitary, while keeping the Dirac equation (see below) covariant under the action of the full Lorentz group: 
\beq\label{def:LT}
L_\bT := -i S_2 \cC,\eeq
 where \beq\label{def:Sk}
 S_k := \begin{pmatrix}  \si_k & 0 \\ 0 & \si_k \end{pmatrix},
\eeq
are components of the {\em spin} operator,
and $\cC$ denotes the complex-conjugation operator in $\Aset \cong M_2(\Pset)$, i.e. 
\beq \cC \psi := \begin{pmatrix} \phi_+^* & \chi_-^* \\ 
\chi_+^* & \phi_-^*
\end{pmatrix}.
 \eeq
We now have the projective representation of $O(1,3)$ as the group generated by matrices of the form $L_{\bLa}$ as in \eqref{def:Spingroup} together with $L_\bP$ and $L_\bT$.  

The {\em Dirac opertor} $\slashed{D}$ on $\Rset^{1,d}$ is by definition
\beq
 \slashed{D} := \ga(\p) = \ga^\mu \frac{\p}{\p x^\mu}.
\eeq
Thus for $d=3$ we have
\beq\label{def:Dpm}
\slashed{D} = \left(\begin{array}{cc} 0 &  \mathds{1}_2\p_0 - \siV\cdot\nab \\ \mathds{1}_2\p_0 +\siV\cdot \nab & 0 \end{array}\right) =: \begin{pmatrix} 0 & D_- \\ D_+ & 0\end{pmatrix}.
\eeq
Thus
$$ D_- := \si(\p),\qquad D_+ := \si'(\p)
$$
(since $\p^k = -\p_k$ for $k=1,2,3$.)  Note that $\slashed{D}^2 = \mathds{1}_4\dal$ and $D_+D_-=D_-D_+ = \Id_2\dal$ where $\dal := \p_0^2 - \De_{\Rset^3}$ is the three-dimensional wave operator.

\subsection{Relativistic wave functions and equations}

\subsubsection{Photon wave function and equation}
According to Kiessling \& Tahvildar-Zadeh \cite{KTZ2017}, in $d$ space dimensions the wave function of a single photon is a rank-two bi-spinor field on $\Rset^{1,d}$ which, when viewed as a linear transformation, has trace-free diagonal blocks.  In the case $d=3$, rank-two bi-spinors are the same as {\em Clifford numbers}, i.e. general elements of the algebra $\Aset$, and thus of the form \refeq{def:PhotonWF}, while the trace-free condition implies 
\beq\label{eq:tracecond}
\tr \phi_+ = \tr \phi_- = 0.
\eeq
Here $\tr$ means the matrix trace, or equivalently, twice the scalar part of the quaternion.
\begin{rem}
    The trace condition \eqref{eq:tracecond} projects out the spin-zero sector. An arbitrary element of $\mathbb{A}$ is a mixture of spin-zero and spin-one fields.
\end{rem}
 Thus in three space dimension the wave function of a single photon has four quaternionic components, or 16 complex components, two of which have to be zero because of \refeq{eq:tracecond}:
\beq \label{def:PWF}
\psiPH = \left(\begin{array}{cc} \phi_+ & \chi_- \\ \chi_+ & \phi_- \end{array}\right).
\eeq
Moreover, according to \cite{KTZ2017} the photon wave function satisfies a Dirac-type equation with a projection term:
\beq\label{eq:DirPH}
-i\hbar \slashed{D} \psiPH + \mPH\Pi \psiPH = 0
\eeq
where $\Pi$ is the projection onto diagonal blocks (defined in terms of the projections $\Pi_\pm$ previously defined in \refeq{def:projections}):
\beq\label{def:projs}
\Pi \Psi := \Pi_+ \Psi \Pi_+ + \Pi_- \Psi \Pi_-,\qquad \forall \Psi \in \Aset,
\eeq
$\hbar$ is Planck's constant, and $\mPH>0$ a dimensional constant to be determined. (The speed of light has been set equal to one.)

The equation \refeq{eq:DirPH} is the Euler--Lagrange equation for an action functional with the real scalar Lagrangian density function given by
\beq \label{Lagrangian}
\ell_{\mbox{\tiny ph}}=  
 \frac{\hbar}{16\pi i} \tr\left({\psiPHb} \gamma^\mu_{} \p_\mu^{} \psiPH - \p_\mu \psiPHb \ga^\mu \psiPH\right) 
+ \frac{\mPH}{8 \pi} \tr \left( \psiPHb \Pi \psiPH \right).
\eeq
\begin{rem}
The massive version of \eqref{eq:DirPH}, i.e. where the projection operator is replaced by the identity, was already written down by M. Riesz a long time ago \cite{Rie1946}:
\beq\label{eq:MassiveSpin1}
-i\hbar \slashed{D} \psi + m \psi = 0,\qquad \psi\in M_4(\Cset).
\eeq
This is the Euler--Lagrange equation for an action functional with real scalar Lagrangian density given by
\beq \label{LagrangianMassive}
\ell=  
 \frac{\hbar }{16\pi i} \tr\left({\psib} \gamma^\mu_{} \p_\mu^{} \psi - \p_\mu \psib \ga^\mu \psi\right) 
+ \frac{m }{8 \pi} \tr \left( \psib  \psi \right).
\eeq
See the Appendix at the end of this paper for a brief account of the genesis of the photon wave function, 
and its precursors in the works of Riesz and Harish-Chandra. For a fuller account, and a review of all the previous attempts at constructing a photon wave function and wave equation, see \cite{KTZ2017}.
\end{rem}

\subsubsection{Consequences of the photon wave equation}

In terms of the wave function's quaternionic constituents, \refeq{eq:DirPH}  can be written as 
\beq\label{eq:masslessDir}
\left\{\begin{array}{rcl} 
D_\pm \chi_\mp & = & -\frac{i\mPH}{\hbar} \phi_\mp \\
D_\pm \phi_\pm & = & 0.
\end{array}\right.
\eeq
Note that the above implies that {\em all} components of $\psiPH$ satisfy the massless linear Klein-Gordon (a.k.a. classical wave) equation:
\beq\label{eq:waves}
\dal \phi_\pm = 0,\qquad \dal \chi_\pm = 0.
\eeq

Let
$$\rho: O(1,3) \to SL(2,\Cset)\times SL(2,\Cset),\qquad \rho(\bLa) = L_{\bLa}$$
denote the projective representation of the Lorentz group given in the above. It is easy to check that \refeq{eq:DirPH} is covariant with respect to the action of the full Lorentz group, i.e., given $\psiPH(x)$ any solution of \refeq{eq:DirPH} and any $\bLa\in O(1,3)$, the spinor

\beq\label{eq:lorcov2}
\psiPH'(x) := L_{\bLa} \psiEL( \bLa^{-1}x )L_{\bLa}^{-1}
\eeq
is also a solution of \refeq{eq:DirPH}.
This implies that the components $\phi_\pm,\chi_\pm$ of $\psiPH$ transform in the following way: For $\bLa\in SO_0(1,3)$ we have
\beq\label{eq:lorpropPH}
\phi_+ \to A\phi_+\tilde{A},\qquad \phi_- \to A^* \phi_-A^\dag,\qquad \chi_+\to A^*\chi_+\tilde{A},\qquad \chi_- \to A\chi_-A^\dag
\eeq
where $A\in SL(2,\Cset)$ is such that $L_{\bLa} = \left(\begin{array}{cc} A & 0 \\ 0 & A^{-\dagger} \end{array}\right)$ is the projective representation of $\bLa$, while under $\bLa = P$ the space reflection (parity) transformation, using (\ref{eq:lorcov2},\ref{def:LP}) we have
\beq\label{eq:lorparitPH}
\phi_\pm \to \phi_\mp,\qquad \chi_\pm\to \chi_\mp.
\eeq
Finally, under $\bLa = T$ the time-reversal transformation, using (\ref{eq:lorcov2},\ref{def:LT}) we have
\beq\label{eq:lortimePH}
\phi_\pm \to \phi_\pm^*,\qquad \chi_\pm \to \chi_\pm^*.
\eeq
Note that this in particular implies that $\psiPH^* := \begin{pmatrix} \phi_+^* & \chi_-^*\\ \chi_+^* & \phi_-^*\end{pmatrix}$ solves the time-reversed version of the equation \refeq{eq:DirPH}.

Let $\accentset{\circ}{\psiPH} = \left(\begin{array}{cc} \accentset{\circ}{\phi_+} & \accentset{\circ}{\chi_-}\\ \accentset{\circ}{\chi_+} & \accentset{\circ}{\phi_-}\end{array}\right)$ be initial data supplied on the Cauchy hypersurface $\{x^0 = 0\}$ for \refeq{eq:DirPH} that is subject to compatibility conditions $\tr\accentset{\circ}{\phi}_\pm = \tr\left((\siV\cdot \nab)\accentset{\circ}{\phi}_\pm\right) = 0$. The corresponding initial value problem for $\psiPH$ is equivalent to the following Cauchy problems for the classical wave equation:
\beq\label{eq:WEs}
\left\{\begin{array}{rcl} \dal \phi_\pm  &= & 0 \\
\phi_\pm|_{x^0 = 0} &=& \accentset{\circ}{\phi_\pm} \\
\p_0\phi_\pm|_{x^0 = 0} & = &\mp \siV\cdot\nab \accentset{\circ}{\phi_\pm}\end{array}\right.
\qquad
\left\{\begin{array}{rcl} \dal \chi_\pm  & = & 0 \\
\chi_\pm|_{x^0 = 0} &=& \accentset{\circ}{\chi_\pm} \\
\p_0\chi_\pm|_{x^0 = 0} & = & \mp\frac{i\mPH}{\hbar} \accentset{\circ}{\phi_\pm} \pm \siV\cdot\nab \accentset{\circ}{\chi_\pm}\end{array}\right.
\eeq
\subsubsection{The diagonal blocks of the photon wave function}
We now establish that the diagonal blocks $\Pi\psiPH$ of the photon wave function propagate only in transversal modes:
As explained in \cite{KTZ2017}, by fixing a Lorentz frame for the Minkowski space, one can find $\be_\pm, \bb_\pm: \Rset^{1,3}\to \Rset^3$ such that
\beq\label{eq:insidepsi}
\phi_+ = i \siV \cdot (\be_+ + i \bb_+),\qquad \phi_- = - i \siV \cdot (\be_- - i \bb_-).
\eeq
Thus setting $\fV_\pm := \be_\pm + i \bb_\pm$, in this frame
the equation for the diagonal blocks become
\beq 
\begin{array}{lcl}\label{DIRACeqSPLITdiag}
\phantom{\hbar} (\p_t + \siV\cdot\nab)(\siV\cdot\fV_+) &=& 0, \\
\phantom{\hbar} (\p_t - \siV\cdot\nab)(\siV\cdot\fV_-^*) &=& 0.
\end{array}
\eeq
Now, it is well-known \cite{LapUhl1931,OppiPHOTON,IBBphotonREV} that the equation
\beq \label{eq:Weyl}
(\p_t + \siV\cdot\nab) (\siV\cdot\fV) = 0
\eeq
for $\fV :=  \be + i \bb:\Rset^{1,3}\to\Cset^3$
is \emph{formally} equivalent to the Maxwell system of equations for a source-free electric field $\be$ and a 
magnetic induction field $\bb$ in the given Lorentz frame, viz.
\beq \label{MaxEB}
\begin{array}{ll}
&\p_t\be - \nab\times\bb = 0, \qquad \nab\cdot \be = 0,\\
&\p_t\bb + \nab\times\be = 0, \qquad \nab\cdot \bb = 0.
\end{array}
\eeq
 Thus, the first equation in (\ref{DIRACeqSPLITdiag}) and the second equation
in (\ref{DIRACeqSPLITdiag}) each {\em separately} are equivalent to (\ref{MaxEB}).
 This proves the absence of longitudinal modes in (\ref{DIRACeqSPLITdiag}). It also establishes that to describe a photon wave function fully, even in a fixed frame, a {\em pair} of Maxwell fields (solutions of \eqref{MaxEB}) are needed, one will not be enough.

\subsubsection{The off-diagonal components of the photon wave function.}\label{sec:offdiag}

More generally, it was established in \cite{KTZ2017} that 
\beq
\chi_+ = \si'(\bfa_+),\qquad \chi_- = \si(\bfa_-),
\eeq
where $\bfa_\pm$ are complex-valued 1-forms on the configuration Minkowski space. 
Even though this  appears to imply that the off-diagonal terms have 8 complex, or 16 real degrees of freedom, 
it turns out that half of those are due to {\em gauge} freedom, which for \refeq{eq:DirPH} consists of 
\begin{equation}\label{eq:DIRACmZEROeqnGAUGE}
    \psiPH \mapsto \psiPH + (\Id - \Pi)\Upsilon,\qquad \slashed{D}\Upsilon = 0.
\end{equation}
 More precisely, the following was shown in \cite{KTZ2017}:
\begin{prop*}
Let $\psiPH = \left(\begin{array}{cc} \psi_+ & \chi_-\\ \chi_+ & \psi_- \end{array} \right)$ be a solution of \refeq{eq:DirPH}. 
 There exists a gauge transformation \refeq{eq:DIRACmZEROeqnGAUGE},
 such that after applying it, the $\chi_\pm$ are Hermitian matrices (equivalently, 
the $\bfa_\pm$ are real-valued.)
\end{prop*}
Fixing a Lorentz frame, we can therefore set
\beq\label{eq:insidechi}
\chi_+ = \varphi_+\si_0 - \siV \cdot \ba_+,\qquad \chi_- = \varphi_- \si_0 + \siV \cdot \ba_-,
\eeq
for $\varphi_\pm:\Rset^{1,3}\to \Rset$ and $\ba_\pm:\Rset^{1,3}\to \Rset^3$.
Let us also denote the components of $\phi_\pm$ as in \eqref{eq:insidepsi}. Let $\psiPH$ be a solution of \refeq{eq:DirPH}. 
 Writing \eqref{eq:DirPH} out in components we then obtain
\beq\label{eq:potentials}
\p_t \varphi_\pm + \nab \cdot \ba_\pm = 0,
\qquad  
\be_\pm = \tfrac{\hbar}{\mPH } \left(-\nab \varphi_\pm - \p_t \ba_\pm\right),
\qquad 
\bb_\pm = \tfrac{\hbar}{\mPH } \nab\times \ba_\pm.
\eeq
It thus appears that the relationship of the off-diagonal terms in the photon wave function 
 $\psiPH$ to its diagonal terms is {\em formally} the same as that of {\em electromagnetic potentials} 
(in Lorenz gauge) to their corresponding electromagnetic fields\footnote{We stress that this is only a formal, mathematical correspondence.  The wave function of any quantum mechanical system of particles must be defined on the {\em configuration space} of those particles, a space whose dimension goes up with the number of particles, while classical fields such as the electric and magnetic fields are defined on {\em physical space}, which always has the same dimension regardless of how many particles are in it.}.

Equations \refeq{DIRACeqSPLITdiag} also imply that $\varphi_\pm$ and $\ba_\pm$ must satisfy the classical wave equation: 
\beq\label{eq:wavespot}
 \p_t^2 \varphi_\pm - \Delta \varphi_\pm = 0,\qquad  \p_t^2 \ba_\pm - \Delta \ba_\pm = 0.
\eeq

\subsubsection{Photon probability current}
The existence of a conserved probability current is of profound importance to the understanding of the dynamics of a quantum particle. In \cite{KTZ2017} we showed that the photon wave function $\psiPH$ described above has an intrinsically-defined conserved probability current, one which we constructed in two steps: First we showed that given any Killing field $X$ of Minkowski space, the manifestly covariant current
\beq\label{def:PHcurr}
j_X^\mu := \frac{1}{4} \tr \left(\overline{\phiPH} \ga^\mu \phiPH \ga(X)\right),
\eeq
where $\phiPH := \Pi \psiPH$, is conserved, i.e. 
\beq\label{eq:PHconserv}
\p_\mu j_X^\mu = 0.
\eeq
Here $\overline{\psi} := \ga^0 \psi^\dag \ga^0$ is the Dirac adjoint for rank-two bispinors, and $\ga(X) := \ga_\mu X^\mu$. 

The proof of \refeq{eq:PHconserv} relies on the fact that for $\psiPH$ satisfying the photon wave equation \refeq{eq:DirPH}, its projection onto the diagonal blocks $\phiPH$ will satisfy the massless Dirac equation, which is the same as the {\em massless} version of Riesz's equation \eqref{eq:MassiveSpin1}, for which one has the conservation laws \eqref{riesztes}.  We then set $j_X^\mu := \btau^{\mu}_{\nu} X^\nu$.

We next proved that when $X$ is causal  and future-directed, then so is $j_X$, i.e. $\eta(j_X,j_X)\geq 0$, and $j_X^0 \geq 0$.  We then showed that there exists a distinguished, constant (and therefore Killing) vectorfield $X$ that is completely determined by the wave function $\psiPH$ (in fact given a Cauchy surface $\Si$, it depends only on the initial value of $\psiPH$ on $\Si$.)  We defined $X$ by considering any given Lorentz frame $\{\be_{(\mu)}\}$ of Minkowski space, and defining 
\beq
\pi_{(\mu)} := \int_{\Rset^d} j_{\be_{(\mu)}}^0 dx.
\eeq
Conservation law \refeq{eq:PHconserv} then implies that the $\pi_{(\mu)}$ are constant, while the definitions of $\pi_{(\mu)}$ and $j_X$ imply that as a four-component object, $ \boldsymbol{\pi} := (\pi_{(\mu)})$ transforms correctly, i.e. like a Lorentz 4-vector. Moreover, $\bpi$ is a future-directed causal vectorfield, and is {\em typically} timelike.   We set 
$$X := \boldsymbol{\pi}/|\bpi|^2,\qquad |\bpi|^2 :=\eta(\bpi,\bpi)$$  
and defined  the photon probability current $j_{\mbox{\rm\tiny ph}}$ to be $j_X$ for this particular choice of constant vectorfield $X$. 

 Finally, we showed that $j_{\rm{ph}}$ satisfies the appropriate generalization of the Born rule, i.e., if one defines 
\beq
\rho_{\mbox{\rm\tiny ph}} := j_{\mbox{\rm\tiny ph}}^0,\quad v_{\mbox{\rm\tiny ph}}^k := \frac{j_{\mbox{\rm\tiny ph}}^k}{j_{\mbox{\rm\tiny ph}}^0}
\eeq
then one has the continuity equation
\beq
\p_t \rho_{\mbox{\rm\tiny ph}} + \p_k (\rho_{\mbox{\rm\tiny ph}} v_{\mbox{\rm\tiny ph}}^k) = 0,
\eeq
and moreover, in the Lorentz frame where\footnote{We note that there was a slip of pen in eq. (7.15) of \cite{KTZ2017}.} $X = (1/|\boldsymbol{\pi}|,0,0,0)^T$,  one has
\bna\label{eq:Bornrule}
j^0_{\mbox{\rm\tiny ph}} & = &C_\phi \tr\left(\phiPH^\dag \phiPH\right)  = C_\phi \tr\left(\phi_+^\dag \phi_+ + \phi_-^\dag \phi_-\right),\\
j^k_{\mbox{\rm\tiny ph}} & = &  C_\phi\tr\left( \phiPH^\dag \al_k \phiPH\right) =C_\phi \tr\left( \phi_+^\dag \si_k \phi_+ - \phi_-^\dag \si_k \phi_-\right).
\ena
where
\beq
C_\phi := \frac{1}{\int_{\Rset^d} \tr(\phiPH^\dag \phiPH) dx}
\eeq
is a normalization {\em constant} (i.e. it's time-independent) that depends only on the initial values of the photon wave function.
Therefore, $\rho_{\mbox{\rm\tiny ph}} := j^0_{\mbox{\rm\tiny ph}}$ is for all practical purposes a probability density, which (for normalized wave functions with $C_\phi = 1$) depends quadratically on the wave function\footnote{Note that this is the same situation as in the standard Born rule for non-relativistic quantum mechanics, namely, that the Schr\"odinger wave function needs to be normalized so that its $\rho := \psi^\dag \psi$ integrates to one, and the probability current is quadratic in the {\em normalized} wave function.}.

\section{Photons in Curved Spacetime}
Let $(\mathcal{M},g)$ be a 4-dimensional smooth, connected, orientable and time-orientable Lorentzian manifold.  We are going to additionally assume that $(\cM,g)$ has a Cauchy hypersurface $\Sigma$, i.e. a complete spacelike submanifold of codimension 1 with the property that every inextendible past-directed timelike curve in $\cM$ intersects it at one and only one point.
In other words, we assume that $(\M,g)$ is {\em globally hyperbolic.}

 At every point $p\in \cM$ the tangent space $T_p\cM$ is a copy of the Minkowski space, with $g_p$ a quadratic form of signature $(+,-,-,-)$ on it.  Therefore virtually all of the  above constructions on Minkowski space extend to the tangent spaces at every point of $\cM$.  In particular, the Clifford algebra $\Aset$ has a basis generated by $g$-dependent Dirac matrices $\gamma^\mu$ satisfying
\beq
\gamma^\mu\ga^\nu+\ga^\nu\ga^\mu = 2 g^{\mu\nu}\Id_4.
\eeq
Since $\Aset \cong M_2(\Pset)$, a general element $G \in \Aset$ can be written as
\beq\label{genG}
G = \begin{pmatrix} g_+ & h_- \\ h_+ & g_- 
\end{pmatrix},\qquad h_\pm,g_\pm \in \Pset
\eeq
The photon wave function $\psiPH$ is defined to be an $\Aset$-valued field with trace-free diagonal blocks, as in \eqref{def:PWF}.  It is assumed to satisfy the equation
\beq\label{eq:curvedDirPH}
-i\hbar \slashed{D} \psiPH + \mPH\Pi \psiPH = 0,
\eeq
where $\slashed{D}$ is the Dirac operator of $(\cM,g)$, i.e.
\beq
 \slashed{D} := \ga^\mu D_\mu.
\eeq
Here $D_\mu$ is the {\em covariant spin connection} on $\cM$, which is a lifting of the Levi-Civita connection of the metric $g$ to the spin bundle over $\cM$.  It follows that both the metric $g$ and the $\gamma^\mu$ are covariantly constant with respect to $D$ differentiation.  The projection $\Pi$ is defined as before \eqref{def:projs}.

\subsection{Conservation Laws}
Even though the domain $(\cM,g)$ is not assumed to have any symmetries, one may still wonder if there are conservation laws for \eqref{eq:curvedDirPH} associated with symmetries of the {\em fiber} $\Aset$.  Because $\Aset$ is an algebra, it acts on itself, and since $\dim_\Cset \Aset = 16$ there could in theory be up to a $16$-dimensional set of symmetry generators, with a conserved Noether current associated to each one.  It however turns out that the dynamics generated by \eqref{eq:curvedDirPH} preserves at most 5 of these symmetries. In particular  the projection term in \eqref{eq:curvedDirPH} is not invariant under the full symmetry group.  For a {\em massive} spin-one field satisfying \eqref{eq:MassiveSpin1}, where the projection $\Pi$ in \eqref{eq:DirPH} is replaced with the identity operator, there will be a larger symmetry group, and the resulting conservation laws in fact include those found by Riesz \cite{Rie1946}. We will follow the general procedure outlined in \cite{Chr2000} for finding Noetherian conservation laws that are due to the symmetries of the fiber, for Lagrangian field theories defined on sections of a vector bundle:

Let $G$ be a section of the bundle $\cB$ of rank-2 bi-spinors over the configuration Minkowski space, and let $Z$ denote the generator of the following right 
action 
of $G$ on rank-two bi-spinors $\psi \in \cB$:
\beq 
\psi_s := \psi e^{isG},\qquad s \in (-\ep,\ep).
\eeq

$Z$ is thus a section of the bundle $\mathcal{E} := \cup_{x\in \mathcal{M}} \mathcal{L}(\mathcal{B}_x,\mathcal{B}_x)$ where $\mathcal{B}_x := \pi^{-1}_{\mathcal{B},\mathcal{M}}(x)$ is the fiber in $\mathcal{B}$ over $x\in \mathcal{M}$, and
\beq 
Z \cdot \psi = \left.\frac{d}{ds}\right|_{s=0} \psi_s = i \psi G.
\eeq
Let $\epsilon = \epsilon[g]$ denote the volume form of $\cM$ with respect to the metric $g$, let the Lagrangian density function $\ell_{\mbox{\tiny ph}}$ be defined as in \eqref{Lagrangian}, except with $\p_\mu$ replaced by the spin connection $D_\mu$, and let $\cLph := \ell_{\mbox{\tiny ph}} \epsilon$ be the 4-form corresponding to it.   According to \cite{Chr2000}, the conserved Noether current corresponding to $Z$ is the 3-form $$J^{\mu\nu\la} = p^{\mu\nu\la}_a (Z\cdot\psi)^a$$ where $p = (p^{\mu\nu\la}_a)$ are 
the {\em canonical momenta} (i.e. the derivative of the Lagrangian density $\cLph$ with respect to the {\em canonical velocities} $v = (D_\mu \psi^a)$.)
Furthermore, we have (see \cite[Chap. 4, (1.148)]{Chr2000})
$$
dJ = \mathcal{L}_Z \cLph + ((D Z)\cdot\psi)\wedge p.
$$
It follows that $J$ is a conserved Noether current if the right-hand-side of the above vanishes on solutions $\psi$ of the Euler-Lagrange equations for the action with Lagrangian $\cLph$.

Let the vectorfield $j^\mu = \epsilon^{\mu\nu\la\ka}J_{\nu\la\ka}$ denote the Hodge dual of $J$ with respect to the volume form $\ep[g]$. If $dJ = 0$, it follows that $j$ is divergence-free. Thus defining
\beq\label{def:jG}
j_G^\mu := \left(\overline{\psiPH} \gamma^\mu \psiPH G\right)_S
\eeq
(using the notation \eqref{def:Spart},) we have
 \beq\label{cons}
 \nabla_\mu j_G^\mu = 0
 \eeq
  when $\psiPH$ satisfies \eqref{eq:curvedDirPH}.

From \eqref{cons} it follows that if one introduces coordinates $(t,\bs)$ in a neighborhood of the Cauchy hypersurface $\Sigma \subset \cM$ such that $t$ is a time function on that neighborhood with level sets  $\Sigma_t$ and the original Cauchy hypersurface $\Sigma = \Sigma_0$, then by the divergence theorem one has 
$$
\frac{d}{dt} \int_{\Sigma_t} j_G^0 d^3s = 0.
$$

For the photon Lagrangian \refeq{Lagrangian} we have
\beq \label{helicityCURRENT}
j^\mu = (*p)^\mu_a (Z\cdot\psiPH)^a = \frac{\p \ell_{\mbox{\tiny ph}}}{\p(D_\mu \psiPH)}(i\psiPH G) =  \tr\Big[\frac{1}{16 \pi i} \psiPHb \ga^\mu (i \psiPH G)\Big] =
 \frac{1}{16 \pi}\tr\left(\psiPHb\ga^\mu \psiPH G\right).
\eeq
We look for conditions on $G$ such that the current $j$ is conserved, i.e. 
\beq\label{conds}
\mathcal{L}_Z \cLph[\psiPH] + ((D Z)\cdot\psiPH)\wedge p = 0,
\eeq
for $\psiPH$ satisfying \eqref{eq:DirPH}.  This does not seem to be an easy task, and may require $G$ to be $\psiPH$-dependent.  If one is however content with deriving sufficient conditions on $G$, one can simplify matters considerably by assuming that $G$ is covariantly constant, i.e. 
\beq\label{covconst}
D_\mu G = 0.
\eeq    
(We will investigate nontrivial solutions of this in a moment.) This implies that $D Z = 0$.
We next compute
$$
\cL_Z\ell_{\mbox{\tiny ph}}(\psiPH) = \left.\frac{d}{ds}\right|_{s=0} \ell_{\mbox{\tiny ph}}(\psiPH e^{isG}) = I + II,
$$
where $I$ contains the contribution of the first term in the Lagrangian, the one involving $D\psiPH$, and $II$ has the contribution of the projection term.  We find
$$
I = \frac{\hbar}{16 \pi} \tr \left( - \overline{G} N + N G\right),\qquad N:= \psiPHb \ga^\mu D_\mu \psiPH - D_\mu \psiPHb \ga^\mu \psiPH.
$$
Thus we have $I = 0$ for all $\psiPH$
provided 
\beq \label{Gselfadj}\overline{G} = G
\eeq
(since $\tr( NG) = \tr (GN)$.)  Assuming \eqref{Gselfadj}, we then have
 $$
 II  = \frac{\mPH i}{8\pi} \tr\left( \psiPHb (\Pi (\psiPH G) - (\Pi\psiPH)G)\right) = \frac{\mPH i}{8\pi}\tr\left[ (\phi_-^\dag\chi_- - \chi_-^\dag \phi_- ) h_+ + (\phi_+^\dag\chi_+ - \chi_+^\dag \phi_+ ) h_- \right],
$$
where we have used \eqref{def:PWF} and \eqref{genG}.  Therefore $II=0$ for all $\psiPH$ provided we further restrict $G$ such that
\beq\label{hscalar}
 h_\pm \in \Cset.
\eeq 
Incidentally, we can also see that the Lagrangian \eqref{LagrangianMassive} of the {\em massive} version of the theory has a larger symmetry group, since the term corresponding to $II$ in that case will be zero without the need for assumption \eqref{hscalar}.  In particular, on Minkowski space and for the massive case studied by Riesz in \cite{Rie1946}, one can set $G = \ga^\nu$ for $\nu = 0,\dots,3$, in which case one recovers the conservation laws of the {\em Riesz tensor}\footnote{Riesz does not mention Noether or refer to her work in his paper.  He derives his conservation laws using the standard {\em multiplier method}, i.e. multiplying the equation through by something and integrating by parts.  This led us in \cite{KTZ2017} to erroneously state that Riesz's conservation laws are non-Noetherian. We are glad to have the opportunity here to correct this slip of mind.}:
\beq\label{riesztes}
\nab_\mu \btau^{\mu\nu} = 0,\qquad \btau^{\mu\nu} :=  \frac{1}{4}\tr( \psib \ga^\mu \psi \ga^\nu),
\eeq
whenever $\psi$ satisfies \eqref{eq:MassiveSpin1}. To connect back with what was shown in \cite{KTZ2017}, we observe that if $\psiPH$ satisfies \eqref{eq:DirPH}, then $\Pi\psiPH$ will satisfy the massless Dirac equation, which is \eqref{eq:MassiveSpin1} with $m=0$.  Thus, \eqref{riesztes} holds with $\psi = \Pi\psiPH$, and this is the conservation law used in \cite{KTZ2017} to define a probability current for the photon wave equation on Minkowski spacetime.

The assumption \eqref{covconst} allowed us to find several conserved quantities, but the question is, on a curved spacetime, does it have any nontrivial solutions? Recall that $G$, being a section of the bundle $\cB$, at every point $x\in\cM$ has an expansion in the basis \eqref{basis} for the Clifford algebra $\mathbb{A}$.  It follows that there are $p$-forms $\{G^{(p)}\}$ for $p = 0,1,2,3,4$ such that $G = G^{(0)} \Id + G^{(1)}_\mu \ga^\mu + \dots + G^{(4)} \gamma^5$.  Thus, $G^{(0)}$ is a function and $G^{(4)} = \mathcal{G}\epsilon[g]$ for some function $\mathcal{G}$ on $\cM$.  The condition \eqref{covconst} implies that $G^{(0)}$ and $\mathcal{G}$ are two complex constants.  On the other hand, on a 4-dimensional manifold with no continuous symmetries there are in general no nontrivial covariantly constant 1-forms, 2-forms, or 3-forms, so that $G^{(1)} = G^{(2)} = G^{(3)} = 0$ which, together with the other restrictions \eqref{Gselfadj} and \eqref{hscalar}, on a general spacetime leaves us only with 
\beq\label{crosshel}
G = a \Id + i b \ga^5,\qquad a,b\in \Rset
\eeq
(cf. \cite[(5.10)]{KTZ2017}.) 

If, on the other hand, we are on a Minkowski background, where every vectorfield whose components are constant (in the standard frame given by Cartesian coordinates on $\Rset^{1,3}$) is both a Killing field and is hypersurface orthogonal, \eqref{covconst} just means that the components of $G$ (in the standard Cartesian frame) are constant, and we have the full $5\Cset$- or $10\Rset$-dimensional set of conserved quantities implied by conditions \eqref{Gselfadj} and \eqref{hscalar} on the generator $G$.  Substituting in \eqref{def:jG} for $\psiPH$ in terms of its quaternionic components $\phi_\pm$ and $\chi_\pm$ as in \eqref{def:PWF}, one then obtains the explicit form of the conserved {\em charge}, i.e. the time component $j_G^0$ of the current to be
$$
j_G^0 = \frac{1}{4}\tr \left[ (\chi_-^\dag \phi_+ + \phi_-^\dag \chi_+) g_+ + 
(\chi_+^\dag \phi_- + \phi_+^\dag \chi_-) g_- + (\chi_-^\dag\chi_- + \phi_-^\dag \phi_-) h_+ + (\chi_+^\dag \chi_+ + \phi_+^\dag \phi_+) h_- \right], 
$$
where, by virtue of the assumption \eqref{Gselfadj}, we have $g_+ = g_-^\dag \in \Pset$ and $h_\pm\in \Rset$.
Thus we have 4 complex-valued and 2 real-valued conserved currents.  One of the complex conserved charges, the one corresponding to \eqref{crosshel} had already been identified in \cite{KTZ2017}.  On Minkowski space, using the decompositions \eqref{eq:insidechi} and \eqref{eq:insidepsi} it can be written as
\begin{equation}\label{def:z}
    \mathfrak{z} := \frac{1}{4}\tr (\chi_-^\dag \phi_+ + \phi_-^\dag \chi_+) =\half (\ba_+ \cdot \bb_- - \ba_- \cdot \bb_+) +\frac{i}{2} (\ba_-\cdot \be_+ - \ba_+ \cdot \be_-),
\end{equation}
while the real-valued charges correspond to $h_\pm \in \Rset$, $g_\pm = 0$, and are:
\begin{equation}\label{def:rpm}
    r_\pm := \frac{1}{4}\tr(\chi_\pm^\dag\chi_\pm + \phi_\pm^\dag \phi_\pm) = \half\left( \varphi_\pm^2 + |\ba_\pm|^2\right) + \half\left(|\be_\pm|^2 + |\bb_\pm|^2\right).
\end{equation}
In what follows we are going to focus on three more complex charges, which correspond to $g_+ = g_-^\dag = \bc \cdot \boldsymbol{\sigma}$ for an arbitrary $\bc \in \Cset^3$.  Here $\boldsymbol{\sigma}$ denotes the vector of Pauli matrices (Recall that the Pauli matrices span the subspace of imaginary quaternions, i.e. trace-free Hermitian matrices.)  We thus have a conserved, $\Cset^3$-valued charge
\beq\label{def:J}
\bJ := \frac{1}{4}\tr \left[ (\chi_-^\dag \phi_+ + \phi_-^\dag \chi_+) \boldsymbol{\sigma}\right],
\eeq
that is to say, each component of $\bJ$ is a conserved complex charge. 
Using the decompositions \eqref{eq:insidechi} and \eqref{eq:insidepsi},   one has
\beq \bJ =   -\half\left( \varphi_+\bb_- + \varphi_-\bb_+ + \ba_+\times\be_- +\ba_-\times\be_+\right) - \frac{i}{2}\left( -\varphi_+ \be_- -\varphi_-\be_+ + \ba_+\times\bb_- + \ba_-\times \bb_+\right).\nonumber
\eeq
On a flat background, the quantities $(r_\pm,\mathfrak{z},\bJ)$ together form a (real) ten dimensional set of conserved charges for the photon wave function, in the sense that their integrals over constant $t$ slices $\Si_t$ are independent of $t$.  

For the remainder of this paper, we are going to restrict ourselves to the case of a Minkowski background, and investigate the gauge-dependence of the new conserved charges we have just found.
\subsection{Gauge invariance}
Finding all these new conservation laws may seem significant, but there is a potential snag with the above procedure, namely, we know that the photon wave equation \eqref{eq:curvedDirPH} is invariant under a rather large group of {\em gauge} transformations: it was shown in \cite{KTZ2017} that if one adds to a solution $\psiPH$ of \eqref{eq:curvedDirPH} the off-diagonal part of any bispinor field $\Upsilon$ that solves the massless Dirac equation 
$$
\slashed{D} \Upsilon = 0,
$$
then the resulting wave function $\psiPH+(\Id-\Pi)\Upsilon$ is still a solution, and is physically equivalent to $\psiPH$.  One important consequence of this fact is that the off-diagonal components of $\psiPH$, namely the $\chi_\pm$ can change a lot under a gauge transformation.  For example, already on Minkowski space, $\chi_\pm$ transforms to
\beq\label{gaugetrans}
\chi_\pm' = \chi_\pm + D_\pm v_\pm,\qquad \dal v_\pm = 0
\eeq
with $v_\pm$ two arbitrary solutions of the linear classical wave equation,
while the diagonal blocks $\phi_\pm$ remain invariant.  Thus any expression in which $\chi_\pm$ explicitly appear is not likely to be gauge-invariant.  A conserved quantity that is not invariant under a gauge transformation would be of little value for physical studies.  None of the quantities we have produced so far are manifestly gauge-invariant, since they all have the $\chi_\pm$ appearing in them\footnote{Except in one space dimension ($d=1$), where $\phi_\pm \equiv 0$ and there is no gauge freedom left in $\chi_\pm$.  In that case the $r_\pm$ defined in \eqref{def:rpm} are conserved quantities that were not previously derived.}. 

The situation may however be analogous to the Einstein equations of general relativity, for which it is possible to find a set of conserved currents with densities that are not manifestly invariant under a diffeomorphism, since they depend on various components of the metric written in a specific coordinate system.  And yet in that case one is able to find {\em global} conserved quantites, the so-called ADM quantities, obtained by integrating those densities on the ``sphere at spatial infinity''.  This is in the context of {\em asymptotically flat} spacetimes, i.e. those for which outside a compact set there is a Cartesian coordinate system $(x^\mu) = (x^0,\bx)$ with the property that the metric in those coordinates approaches the Euclidean metric (often at a specified rate) as $|\bx|\to \infty$ for fixed $x^0$.  It is then shown that these integral quantities are both independent of $x^0$ and diffeomorphism invariant. For Einstein's equations the existence of these quantities is connected with the {\em asymptotic} symmetries of such spacetimes, i.e. symmetries at infinity, and there are indeed 10 such quantities, corresponding to the generators of the Poincare group of isometries of Minkowski space, aptly named {\em mass}  (or equivalently {\em energy}), {\em linear momentum, angular momentum}, and {\em center of mass integrals} (see e.g \cite{MPGR1}.)  One also notes that these quantities are obtained by integrating certain surface densities (i.e. 2-forms) on large coordinate spheres that lie in a constant time slice, and then taking the limit as the radius of the sphere goes to infinity.  

For this analogy to work for the photon wave equation on Minkowski background, we need to construct the 2-forms that are to be integrated on coordinate spheres $S_r$ lying in a constant time slice, and we need to verify that the limit as $r\to \infty$ of those integrals is both independent of $t$ and also invariant under our gauge transformations.  Since we have already found several conservation laws for the photon wave equation, we can investigate whether a conserved charge already constructed, say $J_i$, is itself a complete 3-divergence, so that its integral on a large ball of radius $r$ in $\Sigma_t$ can be turned into a surface integral on the large sphere that is the boundary of that ball, so that in the limit as $r\to \infty$ we end up with the total charge of that slice, which we already know is conserved.  We would then need to verify that even though the integrand is not manifestly gauge invariant, the limiting value of the integral is. 

In \cite{KTZ2017} we carried out this program for one of the two conserved charges we had constructed, obtaining an ADM-like quantity that we tentatively termed ``cross helicity''.  Here in this paper we would like to do the same with the vector-valued charge $\bJ\in \Cset^3$.  It turns out however that its components $J_i$ are in fact {\em not} complete divergences.  To remedy this, we will construct a separate set of conserved quantities $K_i$, with the property that the components of $\bL = \bJ - \bK$ {\em are} complete divergences, $L_i = \nabla \cdot \bY_i$, thus arriving at complex-valued ADM-like conserved quantities
$$ 
\mathbf{f}_i := \lim_{r\to \infty} \int_{B_r \subset \Sigma_t} L_i\ d^3s = \lim_{r\to\infty} \int_{S_r = \partial B_r} \bY_i \cdot \bn\ dS.
$$
We will then conclude by showing that adding a solution of the massless Dirac equation to $\psiPH$ does not change $\mathbf{f}_i$ since its contribution to the integrand falls off faster than $1/r^2$.

Finally, viewing this complex vector as $\mathbf{f} = \be + i \bb$ allows us to identify $\mathbf{f}$ with a constant bispinor $\Psi^\infty$ that is {\em self-dual}, i.e. one whose diagonal blocks satisfy $\phi_- = \phi_+^\dag$.

\subsection{Construction of boundary currents}
Boundary current is the name used by Christodoulou \cite[p. 50]{MPGR1} to denote the general procedure of obtaining ADM-like conserved quantities in theories with gauge freedom, such as Einstein's Equations.  

In our setting we need to correct the conserved charges $\bJ$ so that they become boundary currents.  We begin by recalling that on Minkowski space the quaternionic constituents of $\psiPH$ satisfy the massless wave equation \eqref{eq:waves}, since we have $$\slashed{D}^2 =\Id_4 \dal.$$ 


Now recall that $\slashed{D} = \begin{pmatrix} 0 & D_-\\ D_+ & 0\end{pmatrix}$ and thus $\slashed{D}^2 = \begin{pmatrix} D_-D_+ & 0 \\ 0 & D_+D_-  \end{pmatrix}$.  Also from \refeq{eq:masslessDir} it follows that $D_\pm D_\mp \chi_\pm = 0$.
We thus have 
\beq \label{eq:chiwave}
\dal \chi_\pm = \nab^\mu \nab_\mu \chi_\pm = 0.
\eeq

Consider now the 1-form $\kappa = K_\mu dx^\mu$ with
\beq
K_\mu := \frac{i\hbar}{\mPH } \left( \chi_-^\dag \nab_\mu \chi_+ - \nab_\mu \chi_-^\dag \chi_+\right)_S.
\eeq
From \eqref{eq:chiwave} it follows that 
\beq\label{Green}
\nab^\mu K_\mu = 0
\eeq
 i.e. $\kappa$ is another conserved current for $\psiPH$. (In the Riemannian setting \eqref{Green} is sometimes referred to as {\em the second Green identity.})
  In particular, its time component 
\begin{equation}
    \label{def:K0}
K_0 = \frac{i\hbar}{\mPH } \left( \chi_-^\dag \nab_t \chi_+ - \nab_t \chi_-^\dag \chi_+\right)_S
\end{equation}
is a conserved charge.  Clearly, the same is true for 
\beq\label{def:K}
\bK := \frac{i\hbar}{\mPH } \left[ \left( \chi_-^\dag \nab_t \chi_+ - \nab_t \chi_-^\dag \chi_+\right) \boldsymbol{\sigma}\right]_S.
\eeq
We now claim
\begin{prop}
Let $\mathfrak{z}$, $K_0$, $\bJ$, and $\bK$ be as in \eqref{def:z}, \eqref{def:K0}, \eqref{def:J} and \eqref{def:K}, respectively. Let
\beq
\mathfrak{w} := \mathfrak{z} - K_0,\qquad \bL := \bJ - \bK.
\eeq
Then the complex scalar $\mathfrak{w}$ and each component of $\bL$ are boundary currents, i.e., there exist vectorfields $\bX$ and $\bY_i$, $i=1,2,3$ such that 
\beq
\mathfrak{w} = \nab\cdot \bX,\qquad L_i = \nab\cdot \bY_i,\qquad i=1,2,3.
\eeq
\end{prop}
\begin{proof}
Let $D_\pm := \nab_t \Id \pm \boldsymbol{\sigma}\cdot \nab_{\bs}$ as before.  The quaternionic components of $\psiPH$ satisfy equations \eqref{eq:masslessDir}.  Thus we can write
\begin{eqnarray}\label{calc}
\chi_-^\dag \phi_+ & = & \frac{i\hbar}{\mPH} \chi_-^\dag D_- \chi_+ \\
& = & \frac{i\hbar}{\mPH} \left( \chi_-^\dag \nab_t \chi_+ - \chi_-^\dag ( \siV\cdot \nab) \chi_+\right) \nonumber\\
& = & \frac{i\hbar}{\mPH} \left( \chi_-^\dag \nab_t \chi_+ -  \nab\cdot (\chi_-^\dag\siV \chi_+) + (( \siV\cdot\nab)\chi_-)^\dag\chi_+\right) \nonumber\\
& = & \frac{i\hbar}{\mPH} \left( \chi_-^\dag \nab_t \chi_+ -  \nab\cdot (\chi_-^\dag\siV \chi_+) - (\nab_t\chi_-)^\dag \chi_+ + (D_+\chi_-)^\dag\chi_+ \right)\nonumber \\
& = &  \frac{i\hbar}{\mPH} \left(\chi_-^\dag \nab_t \chi_+ - (\nab_t\chi_-)^\dag \chi_+  -  \nab\cdot (\chi_-^\dag\siV \chi_+) \right) -  \phi_-^\dag \chi_+ \nonumber
\end{eqnarray}
It thus follows that
$$
(\chi_-^\dag \phi_+ + \phi_-^\dag \chi_+)_S - K_0 = - \frac{i\hbar}{\mPH}  (\nab\cdot (\chi_-^\dag\siV \chi_+))_S = - \frac{i\hbar}{\mPH} \nab\cdot  (\chi_-^\dag \siV\chi_+)_S
$$
which is a complete divergence. Thus defining
\begin{equation}\label{def:X}
    \bX := -\frac{i\hbar}{\mPH} \left(\chi_-^\dag\siV \chi_+\right)_S 
\end{equation}
and recalling \eqref{def:z} we obtain
\begin{equation}
   \mathfrak{w} := \mathfrak{z} - K_0 = \nab\cdot\bX.
\end{equation}
Let $\bc\in \Cset^3$ be a fixed vector.  Repeating the calculation in \eqref{calc} with $\siV\cdot \bc$ multiplied on the right gives us the desired vector version, i.e.
$$
\bJ - \bK = - \frac{i\hbar}{\mPH} \nab\cdot  \left[ (\chi_-^\dag\siV \chi_+) \siV\cdot\bc\right]_S.
$$
Thus setting
\beq\label{def:Yi}
 \bY_k := - \frac{i\hbar}{\mPH}\left[ (\chi_-^\dag\siV \chi_+) \si_k \right]_S,\qquad k=1,2,3,
\eeq
establishes the claim.
\end{proof}
Let $S_r$ denote the coordinate sphere of radius $r$ lying in the constant-time slice 
$\Sigma_t$ of the globally hyperbolic, asymptotically flat spacetime $(M,g)$.  Let $\psiPH$ be a photon wave function, i.e. a rank two bispinor satisfying \eqref{eq:curvedDirPH}.  Let $\phi_\pm,\chi_\pm$ denote the quaternionic components of $\psiPH$, as in \eqref{def:PWF}, and let the vectorfields $\bX$ and $\bY_j$ be defined as in \eqref{def:X} and \eqref{def:Yi} respectively.  Let
\beq\label{def:fi}
\mu(\psiPH):= \lim_{r\to \infty} \int_{S_r} \bX \cdot \bn dS,\qquad  \mathbf{f}_i(\psiPH) := \lim_{r\to \infty} \int_{S_r} \bY_i \cdot \bn dS
\eeq
define the complex scalar $\muV$ and components of the complex vector $\beff = \be+i\bb$.  It follows that $\muV$ and $\mathbf{f}_i$ are independent of $t$, and therefore constant.  What remains to show is that, even though $\bX$ and $\bY_i$ are not manifestly gauge-invariant, $\muV$ and $\beff$ are.  
\begin{rem}
    Using the decomposition \eqref{eq:insidechi} with respect to the frame $\{\frac{\p}{\p x^\mu}\}$ associated with the Cartesian coordinates on Minkowski space, we obtain 
    \begin{eqnarray}\label{def:eb}
    \mu & = & \frac{\hbar}{2 \mPH} \lim_{r\to \infty} \int_{S_r}\left( \ba_-\times\ba_+ +i(\varphi_- \ba_+ - \varphi_+ \ba_- )\right)\cdot \bn\ dS_r \\
        \bb & = & \frac{\hbar}{2 \mPH} \lim_{r\to \infty} \int_{S_r} (\ba_-\cdot\bn)\ba_+ + (\ba_+\cdot \bn)\ba_- - (\ba_-\cdot \ba_+ + \varphi_-\varphi_+)\bn\  dS_r \\
        \be & = & \frac{\hbar}{2 \mPH} \lim_{r\to \infty} \int_{S_r} (\varphi_+ \ba_- + \varphi_- \ba_+) \times \bn\ dS_r 
    \end{eqnarray}
\end{rem}

\subsection{Example of this construction in Minkowski space}
Here we show how to construct photon wave functions $\psiPH$ with a finite, nonzero bispinor ``at infinity" in Minkowski space.  Using gauge-invariance, we are going to look for solutions to (\ref{eq:potentials}--\ref{eq:wavespot}) that are not just in Lorenz gauge but also in Coulomb gauge.  This implies that without loss of generality $\varphi_\pm \equiv 0$, and yields the following equations:
\begin{equation}\label{eq:coul}
    \dal \ba_\pm = 0,\qquad \nabla\cdot \ba_\pm = 0.
\end{equation}
It is also clear that for the quantities defined in \eqref{def:eb} to be finite, one needs 
\begin{equation}\label{cond:asymp}
    |\ba_\pm| \sim \frac{c_\pm}{|\bs|}\qquad \mbox{as } |\bs| \to \infty.
\end{equation}
Equations \eqref{eq:coul} can be solved exactly using separation of variables, and a complete basis for ``outgoing" solutions with finite $L^2$ norm for $\be_\pm,\bb_\pm$ was found by Marchal \cite{Marchal}.  These have the form $\ba = \sum_{n=1}^\infty \ba_n$ with
\begin{equation}\label{eq:an}
    \ba_n(t,\bs) = \frac{\bP_n(\bs)}{|\bs|^n}g_n(t,|\bs|) + \frac{\bP_{n-2}(\bs)}{|\bs|^{n-2}}h_n(t,|\bs|),
\end{equation}
where $\bP_{-1} = \bP_0 \equiv 0$ and for $n\geq 1$, $\bP_n, \bP_{n-2} \in \Rset^3$ are two vector-valued homogeneous harmonic polynomials of degree $n$ and $n-2$ respectively, with the property that
\begin{equation}
    \bs \cdot \left(\bP_n(\bs) + |\bs|^2 \bP_{n-2}(\bs)\right) = 0,
\end{equation}
\begin{equation}\label{def:gnhn}
    g_n(t,r) := (-r)^n \left(\frac{1}{r}\frac{\p}{\p r}\right)^n \frac{f_n(t-r)}{r},\qquad h_n(t,r) := (-r)^n \left(\frac{1}{r}\frac{\p}{\p r}\right)^n \frac{f_n''(t-r)}{r},
\end{equation}
with $f_n\in C^{n+2}(\Rset)$ arbitrary. To ensure that the $L^2$ condition is satisfied, Marchal took $f_n$ to be compactly supported in an interval not containing zero, but this is clearly not necessary.  We can see that for the asymptotic condition \eqref{cond:asymp} as well as the $L^2$ condition on $\be_\pm,\bb_\pm$ to be satisfied, it is enough that $f_n \in C^{n+2}(\Rset_+)$ with
\begin{equation}\label{cond:f}
    f_n(r) \sim \left\{\begin{array}{ll} r^{n+1} & r \to 0^+\\ r^n & r \to \infty. \end{array}\right.
\end{equation}
Marchal showed that these solutions, when appropriately normalized, form a complete orthonormal basis for the outgoing solutions of \eqref{eq:coul} (For incoming solutions, 
change $t-r$ to $t+r$ in \eqref{def:gnhn}.)

Armed with this decomposition we construct an example of a photon wave function with a nontrivial bispinor at infinity by letting $\ba_+ = \ba_1$ and $\ba_- = \ba_2$, as defined in \eqref{eq:an}, with the choices 
$$\bP_1(\bs) = ( 0, -s_3, s_2),\qquad \bP_2(\bs) = ( -s_2 s_3, s_1 s_3, 0)$$
and with $f_1,f_2$ satisfying the conditions \eqref{cond:f}.  It then follows that $\be = \mathbf{0}$ by the Coulomb gauge condition and, since our choice makes $\ba_+\cdot\bn = \ba_-\cdot \bn = 0$,
$$
\bb = \frac{-\hbar}{2 \mPH} \lim_{r\to \infty} \int_{S_r} (\ba_+\cdot \ba_-)\bn\  dS_r = \frac{\hbar}{2 \mPH} \lim_{r\to \infty} g_1(0,r)g_2(0,r) r^2 \int_{S_1} (s_1 s_3^2)\begin{pmatrix}
    s_1 \\ s_2 \\ s_3
\end{pmatrix}dS_1 = \frac{C\hbar}{ \mPH} \begin{pmatrix} 1 \\ 0 \\ 0 \end{pmatrix}
$$
for some nonzero numerical constant $C$, since by their definitions \eqref{def:gnhn}, $g_1 ,g_2 \sim \frac{1}{r}$ as $r\to \infty$.

\subsection{Gauge-invariance of the asymptotic bispinor}
\begin{thm}
Let 
$
\Upsilon %
$
be a solution of the massless Dirac equation $\slashed{D}\Upsilon = 0$ with smooth compactly supported data, and let 
$$
\psiPH' := \psiPH + (1-\Pi)\Upsilon
$$
be the gauge-transformed version of $\psiPH$. Let $\psi_\pm$ denote the off-diagonal blocks of $\Upsilon$. It follows that 
$$
\phi_\pm' = \phi_\pm,\qquad\chi'_\pm = \chi_\pm +\psi_\pm
,\qquad D_\mp \psi_\pm = 0.
$$
We then have
$$
\mu(\psiPH') = \mu(\psiPH),\qquad\mathbf{f}_i(\psiPH') = \mathbf{f}_i(\psiPH).
$$
\end{thm}

\begin{proof}
Recall that there exist two real-valued 2-forms $\mathfrak{f}_\pm$ and two real-valued 1-forms $\mathfrak{a}_\pm$ such that
$$ \phi_+ =\Sigma(\mathfrak{f}_+),\quad \phi_- = \Sigma'(\mathfrak{f}_-),\quad \chi_+ = \sigma'(\mathfrak{a}_+),\quad \chi_- = \sigma(\mathfrak{a}_-).$$
Moreover, if the above $\psiPH$ satisfies our photon wave equation \eqref{eq:DirPH}, then we will have
$$
d \mathfrak{a}_\pm = \frac{\mPH}{\hbar}\mathfrak{f}_\pm,\qquad \delta \mathfrak{a}_\pm = 0,\qquad d \mathfrak{f}_\pm = 0,\qquad \delta \mathfrak{f}_\pm = 0.
$$
Here $\delta = *d*$ is the co-differential (or divergence) operator, and $*$ is the Hodge star operator with respect to the spacetime metric $g$.

By the same token, if $\Upsilon = \begin{pmatrix}
    0 & \psi_-\\ \psi_+ & 0 
\end{pmatrix}$ satisfies the massless Dirac equation $\slashed{D}\Upsilon=0$, it follows that there are real-valued 1-forms $\al_\pm$ such that
$\psi_+ = \si'(\al_+)$, $\psi_- = \si(\al_-)$, and we have
$$
d\al_\pm= 0, \qquad \delta \al_\pm = 0.
$$
Thus by the Poincare lemma, there are functions $h_\pm$ such that $\al_\pm = d h_\pm$, and moreover we have 
\begin{equation}\label{waveeq}
    \dal_g h_\pm = 0,
\end{equation} i.e. the $h_\pm$ solve the classical wave equation on the background spacetime $(\cM,g)$, which we have assumed to be Minkowski space again.
Thus the Cauchy problem for \eqref{waveeq} with smooth initial data prescribed on $\Si_0$ is uniquely solvable, 
 the solution is smooth, and satisfies the domain of dependence property.  In particular, 
if the data is compactly supported, so would be the solution at all times.  In that case, since $\psi_\pm$ are compactly supported, they won't make a contribution to the boundary currents $\mu$ and $\beff_i$.  
\end{proof}
We note that using standard density arguments, the compact support assumption on the data of $h_\pm$ can be relaxed, to include data belonging to appropriate Sobolev spaces (e.g. $\dot{H}^1$) without changing the conclusion. This will allow the dynamics of the wave function to be well-defined on $L^2$-based Hilbert spaces.

\section{Acknowledgments} 
We thank Prof. Premala Chandra for wonderful conversations, and reminiscences about her father Harish-Chandra.  We are grateful to the anonymous referee for a careful reading of our manuscript and for the truly helpful suggestions on how to improve it.

\appendix
\bigskip
\noindent{\bf \LARGE Appendix}

\section{Brief History of the Photon Wave Equation}\vspace{-.2truecm}
As we remarked earlier, the massive version of our photon wave equation \eqref{eq:DirPH}, i.e. where the projection operator is replaced by the identity, was already written down by M. Riesz a long time ago \cite{Rie1946}:
\beq\label{rie}
-i\hbar \slashed{D} \psi + m \psi = 0,\qquad \psi\in M_4(\Cset).
\eeq
Riesz however did not address the question of the physical significance of $\psi$, beyond the fact that it was a $4\times 4$ matrix-valued field.  His main point was to either obtain new conservation laws for the Dirac equation for the electron, or see the known ones in a new light.  He did this by first finding the conservation laws that hold for \eqref{rie}, and then using the fact that a 4-component spin-half Dirac wave function $\Psi_{el}$ can be viewed as the only nonzero column of a $4\times 4$ matrix $\psi$, and that the set of all such matrices forms a {\em minimal left ideal} in the space of all {\em Clifford numbers} (his term for elements of $\mbox{\bf Cl}(4)$), to restrict those conservation laws back to electron wavefunctions $\Psi_{el}$.    In particular, Riesz seems to be unaware of the fact that a Clifford-number-valued field can represent a spin-one object, and that the equation \eqref{rie} has the interpretation of a relativistic wave equation for a massive spin-one particle.  Indeed, \eqref{rie} is equivalent to Kemmer's equation \cite{Kem1939} for the same type of particle: 
\begin{equation}\label{eq:Kemmer}
\beta^\mu \p_\mu \psi + m \psi = 0,
\end{equation} with $\psi \in \Cset^{10}$ and matrices  $\beta_\mu$ satisfying the relations
$\beta_\mu\beta_\nu\beta_\rho + \beta_\rho\beta_\nu\beta_\mu =  g_{\rho\nu}\beta_\mu + g_{\mu\nu}\beta_\rho$, first found by Duffin \cite{Duf1938}.  

At roughly the same time as Riesz, Harish-Chandra was working with the Kemmer equation, trying to find what the corresponding equation would be for a {\em massless} spin-one particle, and realizing that setting the mass parameter $m$ in \eqref{eq:Kemmer} to zero yields only trivial solutions.  He showed in \cite{HC1946} that instead one must replace the mass term with a {\em projection operator} onto a 6-dimensional subspace of $\Cset^{10}$.  He then went on and found many conservation laws for the ``massless Kemmer equation" he had obtained, but not the ones that would correspond to what Riesz had found. 

With hindsight, for someone who would be aware of these two strands of thought, it seems an easy move to combine them, and come up both with the photon wave equation \eqref{eq:DirPH} as well as its key conserved current \eqref{def:PHcurr}, long before we accomplished this task in 2017, but apparently no one did (although many came close, see \cite{KTZ2017} for details.)  This may be partly because Riesz's paper appeared in a very obscure publication, and also he was generally ignored by physicists who were his contemporaries\footnote{See for example G\aa rding's account of Riesz's years in Lund \cite{GardingHISTORY}.}. He did not seem to succeed in convincing them of the significance of Clifford algebras for modern physics. Riesz's students also do not seem to have picked up on this particular aspect of his work\footnote{During his long career in Stockholm and then Lund, M. Riesz had relatively few students, although this apparent lack of quantity was more than made up by quality, since those students became extremely well-known mathematicians: E. Hille, L. G\aa rding, and L. H\"ormander, just to name three.  Hille in particular is relevant to our story, since he came to the US and trained a whole generation of first-rate analysts at Yale, among them I. Segal, who did the same at Chicago and then at MIT.  W. Strauss was Segal's first student at MIT, and continued this tradition at Brown, where his first student was R. Glassey.  Glassey's first student at Indiana was Tom Sideris.  The second author of this paper is another academic grandchild of Strauss.}.

Meanwhile, two years after Harish-Chandra (who was Dirac's assistant in Cambridge) published his work on Kemmer's equation, he accompanied Dirac to Princeton, 
where he switched to doing pure mathematics\footnote{When asked for the reason for this switch, he is quoted to have said that now he could believe in what he was doing! \cite{Bor2010}} (something for which apparently Dirac had a hard time forgiving him), eventually becoming one of the greatest algebraists of the 20th century, so he did not pursue this question further, either.  In the ensuing decades, the empirical successes of Quantum Electrodynamics and Quantum Field Theory removed the urgency for physicists of addressing foundational questions such as what a photon is and how  it interacts with the electron, so the search for a quantum-{\em mechanical} description of a photon as a particle was relegated to the theoretical backwaters. Various pronouncements by prominent physicists regarding either the futility or the {\em impossibility} of this task, did not help matters either (see \cite[p. 56]{SSL1974}, \cite[p. 3]{WeinbergBOOKqft}.)

\bibliographystyle{plain}

\end{document}